\documentclass[runningheads]{llncs}
\usepackage{microtype}

\usepackage{graphicx}
\usepackage{amssymb}
\usepackage{amsmath}
\usepackage{stmaryrd}
\usepackage{epstopdf}
\usepackage{wrapfig}
\usepackage{subfig}
\usepackage{hyperref}
\usepackage{enumerate}
\usepackage{array}
\usepackage{delarray}
\usepackage{eufrak}
\usepackage{mathbbol}
\usepackage{tikz}
\usepackage{mathabx}
\usepackage{amsmath}
\usepackage{lscape}
\usepackage{changepage}
\usepackage{multicol}
\usepackage{complexity}
\usepackage{todonotes}
\usetikzlibrary{fadings,decorations,decorations.pathreplacing,patterns,arrows}
\newtheorem{notation}{\textbf{Notation}}
\newtheorem{remark*}{Remark}

\newcommand{\N}{\mathbb{N}}
\newcommand{\Z}{\mathbb{Z}}

\renewcommand{\O}{\mathcal{O}}

\title{\mbox{Computational complexity of the avalanche problem}\\ on one dimensional Kadanoff sandpiles}

\titlerunning{Avalanche problem on 1D Kadanoff sandpiles}

\author{Enrico Formenti \inst{1} \and K\'evin Perrot\inst{1,}\inst{2,}\inst{3} \and \'Eric R\'emila\inst{4}}
\authorrunning{E. Formenti, K. Perrot and \'E. R\'emila}

\institute{
Universit\'e Nice Sophia Antipolis - Laboratoire I3S (UMR 6070 - CNRS)\\ 2000 route des Lucioles, BP 121, F-06903 Sophia Antipolis Cedex, France
\and
\small Universit\'e de  Lyon - LIP (UMR 5668 - CNRS - ENS de Lyon - Universit\'e Lyon 1)\\ 46 all\'e d'Italie 69364 Lyon Cedex 7, France
\and
Universidad de Chile - DII - DIM - CMM (UMR 2807 - CNRS)\\ 2120 Blanco Encalada, Santiago, Chile
\and
Universit\'e de  Lyon - GATE LSE (UMR 5824 - CNRS - Universit\'e Lyon 2)\\ Site st\'ephanois, 6 rue Basse des Rives, 42 023 Saint-Etienne Cedex 2, France
\email{enrico.formenti@unice.fr \hspace{1cm} eric.remila@univ-st-etienne.fr\\ kperrot@dim.uchile.cl}
}

\begin{document}

\maketitle

\begin{abstract}
 In this paper we prove that the general {\em avalanche problem} \textbf{AP} is in \NC\, for the Kadanoff sandpile model in one dimension, answering an open problem of \cite{2010-FormentiGolesMartin-KSPMAP}. Thus adding one more item to the (slowly) growing list of dimension sensitive problems since in higher dimensions the problem is \P-complete (for monotone sandpiles).\par\medskip
\end{abstract}

\section{Introduction}\label{s:introduction}
This paper is about cubic sand grains moving around on nicely packed columns in one dimension (the physical sandpile is two dimensional, but the support of sand columns is one dimensional). The Kadanoff Sandpile Model is a discrete dynamical system describing the evolution of sand grains. Grains move according to the repeated application of a simple local rule until reaching a fixed point. 

We focus on the avalanche  problem (\textbf{AP}), 
namely the problem of deciding if adding a single grain 
of sand in the first column of a sandpile given as an 
input causes a series of topples which hit some position
(also given as a parameter).

This is an interesting problem from several points of view. First of all, it is dimension sensitive. Indeed,
it is proved to be \P-complete for sandpiles in dimension 2 or higher~\cite{2010-FormentiGolesMartin-KSPMAP}
and we proved it in \NC$^1$ in this paper. Roughly speaking the problem is highly parallelisable in dimension
$1$ but not in higher dimensions (unless \P=\NC, of course). Second, an efficient solution to this problem
could be useful for practical applications. Indeed, one can use sandpile models for implementing load schedulers 
in parallel computers~\cite{laredo2014}. In this context, answering to \textbf{AP} helps in forecasting the number of 
supplementary processors that are needed to satisfy one more load which is submitted to the system.



\medskip

The paper is structured as follows. Next section introduces the basic notions and results about Kadanoff sandpiles.
Section~\ref{sec:AP} gives the formal statement of \textbf{AP} and
recalls known results about it.
In Section~\ref{sec:ava-peaks-cols}, main lemmata and notions
that are necessary for the proof of the main result are introduced and proved. Section~\ref*{sec:main-result} contains
the main result. Section~\ref{s:conclusion} draws our conclusions and give some perspectives.

\section{Kadanoff sandpile model}
\begin{wrapfigure}{r}{4.4cm}
  \centering \includegraphics{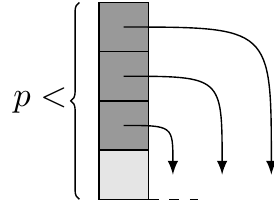}
  \caption{Transition rule with parameter $p=3$.}
  \label{fig:rule}
\end{wrapfigure}
We present the definition of the model in dimension one. A {\em configuration} is a decreasing sequence of integers $h=\,^\omega h_1,h_2,\dots,h_n^\omega$, where $h_i$ is the number of stacked grains ({\em height}) on column $i$, and such that all the heights on the left of $h_1$ equal $h_1$, and on the right of $h_n$ equal $h_n$. Note that all the configurations we consider are {\em finite}. According to a fixed parameter $p$, the transition rule is the following: if the difference of heights between two columns $i$ and $i+1$ is strictly greater than $p$, then $p$ grains can fall from column $i$ and one of them land on each of the $p$ adjacent columns on the right (see Figure \ref{fig:rule}).

A more uniform and convenient representation of a configuration uses {\em slopes}. The {\em slope} at $i$ is the height difference $s_i=h_i-h_{i+1}$. The transition rule thus becomes: if $s_i > p$, then
$$\begin{array}{lcl}
  s_{i-1} &\mapsto& s_{i-1} + p\\
  s_i &\mapsto& s_i - (p+1)\\
  s_{i+p} &\mapsto& s_{i+p} + 1.
\end{array}$$

$|h|=|s|=n-1$ is the {\em length} of the configuration, and the slope of an index $i$ such that $i < 1$ or $n-1 < i$ equals $0$. The transition rule may be applied using different update policies (sequential, parallel, {\em etc}), however we know from \cite{2002-GolesMorvanPhan-linearCFG} that for any initial configuration, the orbit graph is a lattice, hence the stable configuration reached is unique and independent of the update policy. When, from the configuration $s$ to $s'$, the rule is applied on column $i$, we say that $i$ is {\em fired} and we denote $s \overset{i}{\to} s'$ or simply $s \to s'$.

\begin{notation}
  We denote $^\omega s_i$ ({\em resp.} $s_i^\omega$) to say that all the slopes on the left ({\em resp.} right) of column $i$ are equal to $s_i$.
\end{notation}

\begin{notation}
  For any $a,b \in \Z$ with $a \leq b$, let $\llbracket a,b \rrbracket=[a,b] \cap \N$ and $\llbracket a,b)=[a,b) \cap \N$. Finally, $s_{\llbracket a,b \rrbracket}$ denotes the subsequence $(s_a,s_{a+1},\dots,s_b)$.
\end{notation}

A configuration $s$ represented as a sequence of slopes is {\em monotone} if $s_i \geq 0$ for all $i \in \llbracket 1, |s|)$. A configuration is {\em stable} if all its columns are {\em stable}, {\em i.e.}, $s_i \leq p$ for all $i \in \llbracket 1,|s|)$. A stable monotone configuration is therefore a finite configuration $s$ of the form
$$^\omega0,s_1,s_2,\dots,s_{n-1},0^\omega$$

Let $\textbf{gSM}(n)$ be the set of all stable monotone configurations of length $n$ (note that in \cite{2010-FormentiGolesMartin-KSPMAP}, the authors added the restrictive condition $s_i > 0$ for all $i$, whereas we let $s_i \geq 0$ for all $i$ and add the letter $\textbf{g}$ standing for {\em general}). Finally, Let $\textbf{gSM}=\bigcup_{n \in \N} \textbf{gSM}(n)$.

\section{Avalanche problem \textbf{AP}}\label{sec:AP}
An {\em avalanche} is informally the process triggered by a single grain addition on column $1$ (a formal definition is given at the beginning of Section \ref{sec:ava-peaks-cols}). The size of an avalanche may be very small, or quite long, and is sensible to the tiniest change on the configuration. We are interested in the computational complexity of avalanches.\\

\fbox{\parbox{\textwidth-3\parindent}{
  \textbf{Avalanche Problem AP}\\
  A parameter $p \in \N$, with $p \geq 1$, is fixed.\\
  \textit{Instance:} 
  \begin{tabular}[t]{l}
    a configuration $s \in$ \textbf{gSM}\\
    a column $k \in (|s|,|s|+p\rrbracket$
  \end{tabular}\\
  \textit{Question:}
  \begin{tabular}[t]{l}
    does adding a grain on column $1$ trigger a grain addition\\
    on column $k$?
  \end{tabular}
}}\\

For a fixed parameter $p$, the size of the input is in $\Theta(|s|)$. Thanks to the convergence, the answer to this question is well defined and independent of the chosen update strategy.

Let us give some examples. For $p=2$, consider the instance
\[
^\omega0,\underline{2},0,2,1,1,2,1,0,2,0^\omega,
\] 
where the slope of column $1$ is underlined. The question is ``does adding a grain on column $1$ increases the slope of column $k$ equal to $10$ or $11$?'' And the answer is negative in both cases. Here is a sequential evolution:

\[
\begin{array}{rcl}
  ^\omega0,\underline{\textbf{3}},0,2,1,1,2,1,0,2,0^\omega&\to&^\omega0,2,\underline{0},0,3,1,1,2,1,0,2,0^\omega\\
  &\to&^\omega0,2,\underline{0},2,0,1,2,2,1,0,2,0^\omega
\end{array}
\]

For $p=3$, consider the instance $0^\omega,\underline{3},0,2,3,1,3,1,0^\omega$. We have to decide if column $k$ equal to $8$, $9$ or $10$ ends up with a strictly positive slope after a grain is added on column $1$. The answer is positive, positive and negative, respectively. Here is a sequential evolution:
\[
\begin{array}{rcl}
  ^\omega0,\underline{\textbf{4}},0,2,3,1,3,1,0^\omega&\to& ^\omega0,3,\underline{0},0,2,4,1,3,1,^\omega0\\
  \to \,^\omega0,3,\underline{0},0,5,0,1,3,2,0^\omega
  &\to& ^\omega0,3,\underline{0},3,1,0,1,4,2,0^\omega\\
  \to \,^\omega0,3,\underline{0},3,1,0,4,0,2,0,1,0^\omega &\to& ^\omega0,3,\underline{0},3,1,3,0,0,2,1,1,0^\omega
\end{array}
\]

\pagebreak

Known results on the dimension sensitive complexity of \textbf{AP} are the followings.

\begin{itemize}
  \item In dimension one: the restriction of \textbf{AP} to the set of configurations $s$ satisfying $s_i>0$ for all $i$ is known to be in \NC$^1$ \cite{2010-FormentiGolesMartin-KSPMAP}. The key simplification induced by this restriction is the following: an avalanche goes forward if and only if it encounters a slope of value $p$ at distance at most $p$ from the previous one, and thus stops when there are $p$ consecutive slopes strictly smaller than $p$. This condition is not sufficient anymore when we allow slopes of value $0$, as shown for example by the instance $^\omega0,\underline{2},0,2,2,1,2,2,0^\omega$ and $p=2$:
  \[
\begin{array}{rcl}
  ^\omega0,\underline{\textbf{3}},0,2,2,1,2,2,0^\omega&\to&^\omega0,2,\underline{0},0,3,2,1,2,2,0^\omega\\
  &\to&^\omega0,2,\underline{0},2,0,2,2,2,2,0^\omega
\end{array}
\]
  
  \item In dimension two: there are two possible definitions of the model. One has two directions of grain fall, and a configuration is a tabular of sand content that is decreasing with respect to those two directions. In this model \textbf{AP} is \P-complete for all parameter $p>1$ \cite{2010-FormentiGolesMartin-KSPMAP}. The second definition follows the original model of Bak, Tang and Wiesenfeld \cite{1987-BakTangWiesenfeld-SOC}, and it has been proved that information cannot cross (under reasonable conditions) when $p=1$, a strong obstacle for a reduction to a \P-complete circuit value problem \cite{2006-Goles-CrossingInfo2DBTWSandpile}.
  \item In dimension three or greater: sandpiles are capable of universal computation \cite{1996-GolesMargenstern-SPMUniversal}.
\end{itemize}

\section{Avalanches, peaks and cols}\label{sec:ava-peaks-cols}
This subsection partly intersects with the study presented in \cite{lata}, but follows a new and hopefully clearer formulation. For a configuration $s \in \textbf{gSM}$, an avalanche is the process following a single grain addition on column $1$, until stabilization. We will consider avalanches according to the sequential update policy, and prove that it is formed by the repetition (not necessarily alternated) of the following two basic mechanisms:
\begin{itemize}
  \item fire a column greater than all the previously fired columns;
  \item fire the immediate left neighbor of the last fired column.
\end{itemize}

An {\em avalanche strategy} for $s$ is a sequence $a=(a_1,\dots,a_T)$ of columns such that $s^+ \overset{a_1}{\to} \dots \overset{a_T}{\to} s'$, where $s^+$ denotes the configuration $s \in \textbf{gSM}$ on which a grain has been added on column $1$, and $s'$ is stable. Such a strategy is not unique, therefore we distinguish a particular one which we think is the simplest.

\begin{definition}
The {\em avalanche} for $s$ is the minimal avalanche strategy for $s$ according to the lexicographic order, which means that at each step the leftmost column is fired.
\end{definition}

For example, let us consider $p=2$ and the configuration $s=\,^\omega 0,\underline{2},2,2,2,2,0^\omega$, then $(0,2,4,1,3)$ is an avalanche strategy, but {\em the} avalanche for $s$ is $(0,2,1,3,4)$ and leads to the same final configuration thanks to the lattice structure of the model \cite{2002-GolesMorvanPhan-linearCFG}.

Let us give two terms corresponding to the two basic mechanisms underlying the avalanche process, and prove the above mentioned description.
\begin{itemize}
  \item $a_t$ is a {\em peak} $\iff$ $a_t > \max a_{\llbracket 1,t \llbracket}$;
  \item $a_t$ is a {\em col} $\iff$ $a_t=a_{t-1}-1$.
\end{itemize}

First, a simple Lemma.

\begin{lemma}\label{lemma:01}
  An avalanche fires at most once every column.
\end{lemma}
\begin{proof}
  It is straightforward to notice that in order for a column to receive enough units of slope to be fired twice, another column must have been fired twice before, which leads to the impossibility of this situation when adding a single grain on column $1$ of a stable configuration.
  \qed
\end{proof}

Now, the intended description.

\begin{lemma}\label{lemma:peakcol}
  The avalanche of a configuration $s \in \textbf{gSM}$ is a concatenation of peaks and cols.
\end{lemma}
\begin{proof}
  Let $a=(a_1,\dots,a_T)$ be the avalanche for $s$. We prove the lemma by induction on the avalanche size. The first fired column is necessarily $a_1=1$, and we take as a convention that $\max \emptyset = 0$ thus $a_1$ is a peak. Suppose that the result is true until time $t$, we'll prove that $a_{t+1}$ is either a peak or a col. It follows from Lemma \ref{lemma:01} that $a_{t+1} \neq a_t$, and let us denote $a_{t-j}$ with $j \geq 0$ the largest (rightmost) peak before  time $t+1$. The induction hypothesis implies that columns $a_t$ to $a_{t-j}-1$ are cols.
  \begin{center}
    \includegraphics{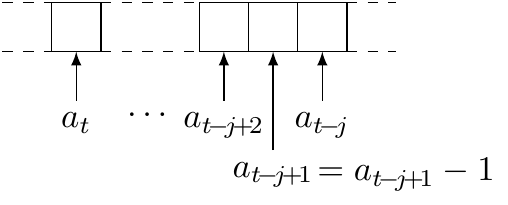}
  \end{center}
  \begin{itemize}
    \item If $a_{t+1} > a_{t}$, by induction on $i$ from $0$ to $j-1$, we have $a_{t+1}>a_t+i$ because $a_t+i$ has already been fired by hypothesis and a column cannot be fired twice (Lemma \ref{lemma:01}). As a consequence $a_{t+1} \geq a_{t-j}$ and for the same reason $a_{t+1} > a_{t-j}$, which was the greatest peak so far, therefore $a_{t+1}$ is also a peak.
    \item If $a_{t+1} < a_{t}$, then, by contradiction, if $a_{t+1} \neq a_t-1$ then the firing at $a_t$ does not influence the slope at $a_{t+1}$, and firing this latter after $a_t$ contradicts the minimality of the avalanche according the lexicographic order, because column $a_{t+1}$ was already unstable at time $t$. Therefore, $a_{t+1}$ is a col.\vspace{-11pt}
  \end{itemize}
  \qed
\end{proof}

Interestingly, avalanches are local processes because they cannot fire a column too far (neither on the left nor on the right) from the last fired column,
as it is proved in the following lemma.

\begin{lemma}\label{lemma:local}
  Let $a$ be the avalanche of a configuration $s \in \textbf{gSM}$,
  $q>0$ \text{ is a peak of } a \text{ implies that } $s_q = p$ \text{ and there exists another peak } $q'$ \text{ satisfying } $q-q' \leq p$.
\end{lemma}

\begin{proof}
  Let $t$ be such that $q=a_t$. By definition of peak, at time $t$ column $q$ could only have received units of slope from columns on its left, that is, by Lemma \ref{lemma:01} it received at most $1$ unit of slope from column $q-p$. Since it was stable on configuration $s$, it has necessarily received this unique unit from column $q-p$ and became unstable thanks to it, which straightforwardly proves both claims.
  \qed
\end{proof}

Note that the converse implication is false. Figure \ref{fig:avalanche} illustrates the results of this section.

\begin{figure}[!h]
  \includegraphics[width=\textwidth]{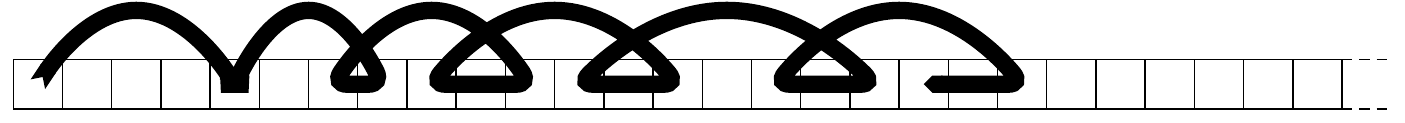}
  \caption{For $p=4$, the arrow pictures the proceedings of an avalanche, which is a concatenation of peaks and cols (Lemma \ref{lemma:peakcol}) where two consecutive peaks are at distance at most $p$ (Lemma \ref{lemma:local}).}
  \label{fig:avalanche}
\end{figure}

\section{\textbf{AP} is in \textbf{NC}$^1$ in dimension one}\label{sec:main-result}
We consider that the input configuration is represented as a sequence of slopes, since it is possible to efficiently transform a representation into another in parallel (for a configuration of size $n$, it requires constant time on $n$ parallel processors). We consider the parameter $p$ as a fixed constant, as it is part of the model definition

\begin{remark*}
In this paper, we consider the parameter p as a fixed constant which is part of the model definition.
Indeed, if p would have been part of the input, which would therefore have size $(|s|+2)\log p$, then comparing the
height of a column to p (in order to know if the rule can be applied at this column) would not take a constant time
anymore. This implies many low level considerations we want to avoid and inflate complexity.
\end{remark*}


We recall that \NC $= \bigcup_{k\in\N}\mathsf{PT/WK}(\log^k n,n^k)$, where $\mathsf{PT/WK}(f(n),g(n))$ (Parallel Time / WorK) is the class of decision problems solvable by a uniform family of Boolean circuits with depth upper-bounded by $f(n)$ and size (number of gates) upper-bounded by $g(n)$, which is more conveniently seen for our purpose as solvable in time $\O(f(n))$ on $\O(g(n))$ parallel PRAM processors. We recall that \NC$^1$=$\mathsf{PT/WK}(\log n, \mathbb R[n])$ where $\mathbb R[n]$ denotes the set of polynomial functions.

As a consequence of Lemmata~\ref{lemma:peakcol} 
and~\ref{lemma:local}, the avalanche process is local. 
Moreover, if we cut the configuration into two parts, we can compute both parts of the avalanche independently, provided a small amount of information linking the two parts. This independency will be at the heart of our construction in order to compute the avalanche efficiently in parallel. Let us have a closer look at how to encode this ``midway information'', which we call \emph{status} (a notion named \emph{trace} has been defined in~\cite{mfcs}, which shares some of those ideas).

For a column $i>p$ of a configuration $s$, the {\em status} at $i$ of the avalanche $a$ for $s$ is the boolean $p$-tuple $(b_0,\dots,b_{p-1})$ such that $b_j=1$ if column $i-p+j$ is fired within $a$, and 0 otherwise. For example, consider the avalanche of Figure \ref{fig:avalanche}, its status at column $8$ (the column where the avalanche starts has index $1$) is $(0,1,0,1)$.

We claim that given a column $i$, the incomplete configuration $s \cap \llbracket i,|s|)$ and the status at $i$ of the avalanche $a$ for $s$, we can compute the avalanche on the part of $s$ that we have, that is, $a \cap \llbracket i, |s|)$.

Note that in the proof of Theorem \ref{theorem:main} we use only simple instances of Lemma \ref{lemma:status}, but we still present it in a general form.

\begin{lemma}\label{lemma:status}
    Given
    \begin{itemize}
      \item a part $s \cap \llbracket i,j)$ with $i+p<j$,
      \item the status at $i$ of the avalanche $a$ for $s$,
    \end{itemize}
    one can compute
    \begin{itemize}
      \item the avalanche on $a \cap \llbracket i, j-p \rrbracket$,
      \item the status of $a$ at $j-p+1$,
    \end{itemize}
    in time $\O(j-i)$ on one processor.\\
    \begin{center}
      \centering \includegraphics{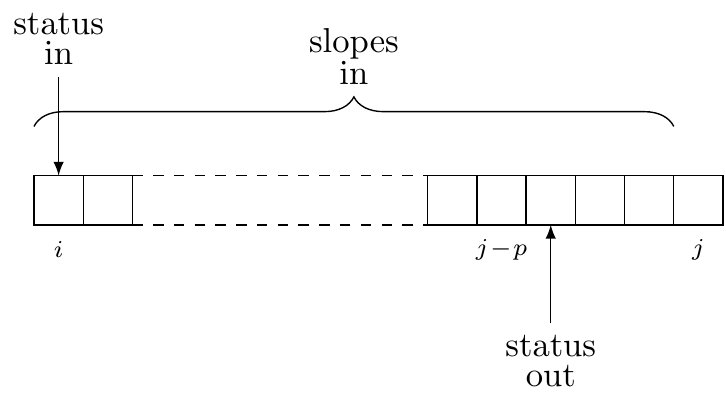}
    \end{center}
\end{lemma}

\begin{proof}
  We claim that given the status of the avalanche $a$ at a column $k$, we can find the smallest (leftmost) peak after column $k$, let us denote it by $q=\min\{q ~|~ q \geq k \text{ and } q \text{ is a peak}\}$, and the part of $a$ between $k$ and $q$, {\em i.e.}, $a \cap \llbracket k,q \rrbracket$. This will be done in constant time thanks to Lemma \ref{lemma:local}: $q-k<p$ so we have to check a constant number of columns. The result then follows an induction on the peaks within $\llbracket i,j )$: from the status at $k$ (initialized for $k=i$), we find the next peak $q$ and compute $a \cap \llbracket k,q \rrbracket$, append it to the previously computed $a \cap \llbracket i,k\rrbracket$, which also allows to construct the status at $q+1$ in constant time. And this process is repeated at most a linear number of times:
  \begin{itemize}
    \item either the avalanche stops at some time,
    \item or the greatest peak encountered is between $j-p$ and $j-1$,
  \end{itemize}
  and in both cases we can compute the intended objects by appending the previously computed parts of the avalanche (Lemma \ref{lemma:peakcol}, recall that the status at $j-p+1$ tells wether columns between $j-p+1-p$ and $j-p$ are fired or not).
  
  Knowing the status of $a$ at $k$, let us explain how to compute the smallest peak after column $k$, denoted $q$, and $a \cap \llbracket k,q \rrbracket$. Let $(b_0,\dots,b_{p-1})$ be the status of $a$ at $k$. From Lemma \ref{lemma:local} the peak $q$ has a value of slope equal to $p$ in $s$ and is at distance smaller than $p$ from $k$. We will now prove that it is very easy to find $q$ in constant time: $q$ is the smallest column $\ell$ such that $0 \leq \ell-k<p$, and $s_\ell=p$ and $b_{\ell-k}=1$.
  \begin{itemize}
    \item Such an $\ell$ is a peak: since $b_{\ell-k}=1$, column $\ell-p$ is fired. When it is fired, it gives one unit of slope to column $\ell$ which can be fired since its slope is initially equal to $p$ and becomes $p+1$. It cannot be a col, which would mean that there is another peak $q''$, greater, which is fired before $\ell$, but from Lemmas \ref{lemma:peakcol} and \ref{lemma:local} this contradicts the minimality of the avalanche because when $q''$ is fired the column $\ell$ is also firable ($\ell-p$ has already been fired since it is at distance strictly greater than $p$ of $q''$).
    \item The smallest peak greater or equal to $k$ satisfies those three conditions: the two first conditions are straightforward from Lemma \ref{lemma:local}. The last condition can be proved by contradiction: suppose there is a peak $q''$ such that $b_{q''-k}=0$, {\em i.e.}, column $q''-p$ is not fired in the avalanche, then $q''$ still needs to receive some units of slope to become unstable, which can only come from its left neighbor $q''+1$ thus this latter has to be fired before it, a contradiction.
  \end{itemize}
  As a consequence of the two above facts, the smallest such $\ell$ is indeed the intended peak $q$, and can be computed in constant time.
  There are $\O(j-i)$ peaks within $\llbracket i,j)$, and each step of the induction needs a constant computation time on one processor, thus the last part of the lemma holds.
  \qed
\end{proof}

Thanks to Lemma \ref{lemma:status} we can perform the computation efficiently in parallel as follows.

\begin{theorem}\label{theorem:main}
  For a fixed parameter $p$ and in dimension one, \textbf{AP} is in \NC$^1$.
\end{theorem}

\begin{proof}
  An input of \textbf{AP} is a configuration $s \in \textbf{gSM}$ and a column $k \in (|s|,|s|+p \rrbracket$. Let $k=|s|+\kappa$ with $\kappa \in (0,p\rrbracket$.
  
  
  The proof works in two stages: first, we compute for every position $i$ the function that associates to each status at $i$, the corresponding status at $i+1$, which we call ``the function status at $i$ $\to$ status at $i+1$'' (since status are elements of $\{0,1\}^p$, the size of these functions is a constant). This can be done in constant time on $|s|$ parallel PRAM using Lemma \ref{lemma:status}; and in a second stage we compute in parallel the function status at $p+1$ $\to$ status at $|s|+1$ using $log |s|$ steps, by pairwise composing the functions as illustrated in Figure \ref{fig:composition}.
  
  \begin{figure}[!h]
    \centering
    \includegraphics[width=\textwidth]{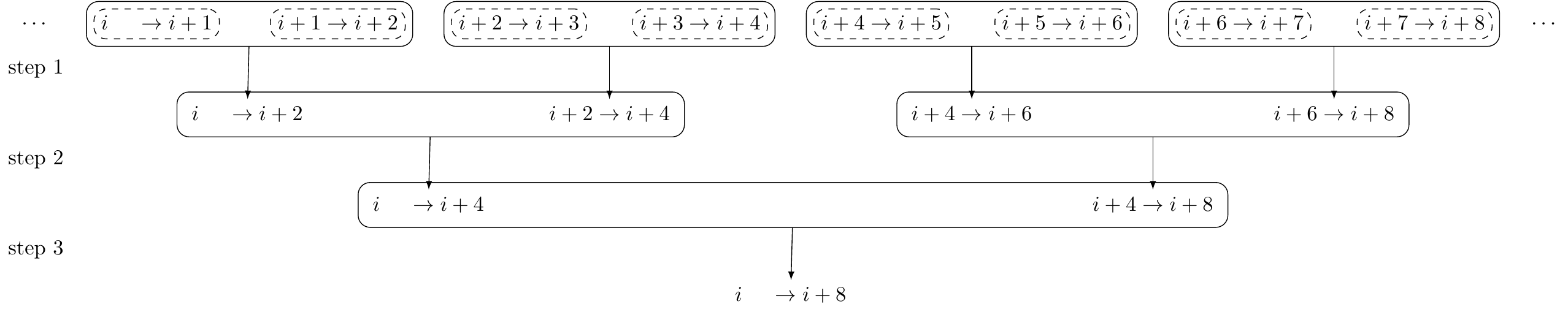}
    \caption{Illustration of the parallel computation, each symbol $x \to y$ represents the function status at $x$ $\to$ status at $y$. Dashed on the top are the functions computed during the first stage. Then, in $log |s|$ steps (each of them uses a polynomial number of parallel processors and a constant amount of time) we compose the circled functions in order to compute the function pointed out with an arrow. This composition is straightforwardly performed in constant time with the two processors: one of them transmits its function to the other one (a function of constant size is transmitted in constant time), and the latter composes two functions of constant size\dots\, in constant time. We perform those computations such that the resulting function $\mu$ has type: status at $p+1$ $\to$ status at $|s|+1$.}
    \label{fig:composition}
  \end{figure}
  
  One of the processors can then finish the job in constant time by first computing the status $\dot{b}$ at $p+1$, which can easily be done in constant time using only $(s_1,\dots,s_{p+1})$ since either $s_1$ is a peak or the avalanche stops, then either $s_{p+1}$ is a peak or the avalanche stops, and finally the cols within $s_1$ and $s_{p+1}$ are straightforwardly found thanks to Lemma \ref{lemma:peakcol}. Then, it computes the status $\ddot{b}=\mu(\dot{b})$ at $|s|+1$, and answers {\em yes} if and only if $\ddot{b}_{\kappa-1}=1$, because columns on the right of $|s|$ cannot be fired but can only receive grains from their left neighbor at distance $p$, so does column $k=|s|+\kappa$.
  
  The complete procedure uses a logarithmic amount of time on a polynomial number of parallel processors (the input has size $\Theta(|s|)$), {\em i.e.}, the decision problem \textbf{AP} is in the complexity class \NC$^1$.
  \qed
\end{proof}

\section{Conclusion and open problem}\label{s:conclusion}

In this paper we proved that \textbf{AP} is in \NC$^1$ in dimension
$1$ solving an open question 
of~\cite{2010-FormentiGolesMartin-KSPMAP}. Going in the direction of~\cite{formenti2004}, one might ask what is
the complexity of \textbf{AP} when the constraint on monotonicity
is relaxed. Clearly, by the results of~\cite{formenti2004},
the problem is in \P, but is it complete?

Another possible generalisation concerns symmetric sandpiles
(see \cite{2006-Formenti-SSPM,2012-EnricoHaDuongTrung-PSSPM,2011-PerrotPhanVanPham-PSSPM}, for example). In this case, the lattice structure of the phase space is lost and therefore we cannot
exploit it in the solving algorithms. This would probably direct the investigations towards non-deterministic computation
and shift complexity results from \P-completeness
to \NP-completeness. 

It also remains to classify the computational complexity of avalanches in two dimensions when the parameter $p$ equals 1 (that is, in the classical model introduced by Bak {\em et al.} \cite{1987-BakTangWiesenfeld-SOC}). As it is exposed in \cite{2006-Goles-CrossingInfo2DBTWSandpile}, this question interestingly emphasizes the links between \NC, \P-completeness, and information crossing.


\section{Acknowledgments}\label{s:ack}
This work was partially supported by IXXI (Complex System Institute, Lyon), ANR projects Subtile, Dynamite and QuasiCool (ANR-12-JS02-011-01), Modmad Federation of U. St-Etienne, the French National Research Agency project EMC (ANR-09-BLAN-0164), FONDECYT Grant 3140527, and N\'ucleo Milenio Informaci\'on y Coordinaci\'on en Redes (ACGO).

%
%
\bibliographystyle{plain}
\bibliography{biblio}




\end{document}